\documentclass[draftcls, 10pt, onecolumn]{IEEEtran}
\usepackage{amssymb, amsmath, graphicx, paralist,subfigure}
\usepackage{amsbsy}
\usepackage{floatflt} 

\usepackage{amsmath}
\usepackage{amssymb}
\usepackage{times}
\usepackage{graphicx}
\usepackage{xspace}
\usepackage{paralist} 
\usepackage{setspace} 
\usepackage{xypic}
\xyoption{curve}
\usepackage{latexsym}
\usepackage{theorem}
\usepackage{ifthen}
\usepackage{subfigure}

{\theoremheaderfont{\it} \theorembodyfont{\rmfamily}
\newtheorem{lem}{Lemma}
\newtheorem{theorem}{Theorem}

}

\begin{document}

\title{\huge DILAND: An Algorithm for Distributed Sensor Localization with Noisy Distance
Measurements}
\author{Usman~A.~Khan$^\dagger$\thanks{$^\dagger$All authors contributed equally to
the paper. This work was partially supported by NSF under grants
\#~ECS-0225449 and~\#~CNS-0428404, and by ONR under grant
\#~MURI-N000140710747.},~Soummya~Kar$^\dagger$,~and~Jos\'e~M.~F.~Moura$^\dagger$\\
            Department of Electrical and Computer Engineering\\
            Carnegie Mellon University, 5000 Forbes Ave, Pittsburgh, PA 15213\\
            \{ukhan, moura\}@ece.cmu.edu,~soummyak@andrew.cmu.edu\\
            Ph: (412)268-7103 Fax: (412)268-3890
}

\maketitle
\begin{abstract}
In this correspondence, we present an algorithm for distributed sensor localization with noisy distance measurements (DILAND) that extends and makes the DLRE more robust. DLRE is a distributed sensor
localization algorithm in~$\mathbb{R}^m~(m\geq1)$ introduced in~\cite{usman_loctsp:08}. DILAND operates when
\begin{inparaenum}[(i)]
\item the communication among the sensors is noisy; \item the
communication links in the network may fail with a non-zero
probability; and \item the measurements performed to compute
distances among the sensors are corrupted with noise.
\end{inparaenum}
The sensors (which do not know their locations) lie in the convex hull of at least~$m+1$ anchors (nodes that know their own locations.) Under minimal assumptions on the connectivity and
triangulation of each sensor in the network, this correspondence shows that, under the broad random phenomena described above, DILAND
converges almost surely (a.s.) to the exact sensor locations.
\end{abstract}

\newpage

\section{Introduction}\label{intro}
Localization is an important problem in sensor networks, not only on its own right, but often as the first step toward solving more complicated and diverse network tasks, which may include environment monitoring, intrusion detection, and routing in geographically distributed communication networks. The problem we consider is when a large number of sensors do not know their locations, only a very few of them know their own. In \cite{usman_loctsp:08}, we presented a distributed sensor
localization (DILOC) algorithm in~$\mathbb{R}^m~(m\geq1)$, when we can divide the~$N$ nodes in the sensor network into these two
sets: the set~$\kappa$ of~$n$ \emph{anchors} where $n\geq m+1$ and the set~$\Omega$
of~$M$ \emph{sensors}, with typically $N\gg n$. The~$n$ anchors are the nodes that know their exact locations, whereas the~$M$ sensors are the nodes that do not know their locations\footnote{In the sequel, we always use this disambiguation for sensors and anchors. When the statement is true for both sensors and anchors, we use the term \emph{node}.}. We assume that the sensors lie in the convex hull of the anchors, i.e., $\mathcal{C}(\Omega)\subset\mathcal{C}(\kappa)$, where $\mathcal{C}(\cdot)$ denotes the convex hull\footnote{The minimal number of anchors required for a non-trivial convex hull in $m$-dimensional ($m$D) space is $m+1$ that is a triangle in $2$D space. We may have more than $m+1$ anchors forming the boundary of a polygon for less stringent requirements on sensor placement, see \cite{usman_allerton:08} for details.}. To each sensor~$l$ in the network, we associate a triangulation set\footnote{In \cite{usman_loctsp:08}, we give a convex hull inclusion test to verify a triangulation set, and we relate the communication radii and the density of deployment to guarantee triangulation with a high probability. We also study the probability of a successful triangulation at each sensor. For this study, we assume a randomly deployed sensor network where the number of sensors in any given area follows a Poisson distribution.}, $\Theta_l$, which is a set of $m+1$ neighboring nodes such that sensor~$l$ lies in their convex hull, i.e., $l\in\mathcal{C}(\Theta_l)$. In DILOC, each sensor,~$l$, updates its location estimate as a linear convex
combination of the estimates of the nodes in its triangulation set, $\Theta_l$, where the coefficients of the linear combination are the barycentric coordinates. Under minimal assumptions on network connectivity, and that the sensor $l$ knows the precise distances in the set
\begin{equation}\label{calD}
\mathcal{D}_l\triangleq\{l\}\cup\Theta_l,
\end{equation}
reference~\cite{usman_loctsp:08} shows that DILOC converges to the exact sensor locations. An interesting contribution of DILOC is that it reduces the centralized nonlinear problem of localization to a linear distributed iterative algorithm under broad assumptions, which can be implemented through local inter-sensor communication in real-time.

Reference~\cite{usman_loctsp:08} extends DILOC and presents a distributed localization algorithm in random environments (DLRE). DLRE is a
stochastic approximation version of DILOC where the DILOC
iterations are weighted with a decreasing weight sequence that
follows a persistence condition. DLRE operates under the following random phenomena:
\textbf{(B.1)} the communication among the sensors and the nodes in their
triangulation set is noisy, i.e., at the~$t$-th iteration,
sensor~$l$ receives only corrupted versions, $y_{ln}^{j}(t)$, $1\leq j\leq m$, of the $m$~components of the neighboring
node~$n$'s state, i.e., $x_{n}^{j}(t)$, given by
\begin{equation}\label{rnd2}
y_{ln}^{j}(t)=x_{n}^{j}(t)+v_{ln}^{j}(t),\qquad n\in\Theta_l,
\end{equation}
where~$\left\{v_{ln}^{j}(t)\right\}_{l,n,j,t}$ is a family of independent zero-mean random variables with finite second moments; \textbf{(B.2)} each inter-node communication link is
 modeled by a binary random
variable, $e_{ln}(t), l\in\Omega, n\in\Theta_l$, which is~$1$ (active link) with probability $q_{ln}$, and $0$ (link failure) with probability $1-q_{nl}$; and \textbf{(B.3)} the local barycentric coordinates, computed at iteration $t$ from the current
noisy distance measurements, can be represented as a perturbation of the exact barycentric coordinates. Reference~\cite{usman_loctsp:08} shows that, under \textbf{(B.1)--(B.3)}  and if the perturbation in~\textbf{(B.3)} is unbiased, DLRE converges to the sensor exact locations; however, if the perturbation in~\textbf{(B.3)} is biased,  DLRE converges with a steady-state error (bias).

In this correspondence, we modify DLRE to
present the algorithm DILAND (distributed sensor localization with
noisy distance measurements) and show that it converges a.s. to the exact sensor locations under much broader distance
measurement noise assumptions. In DILAND, we replace \textbf{(B.3)} above with the following weaker condition: $\overline{\mbox{\textbf{(B.3)}}}$ at every
iteration $t$, we assume there exist computationally efficient
estimates of the required inter-sensor distances based on all
distance measurements till time $t$ such that these estimates
are consistent, i.e., they converge a.s.~to the exact distances
as $t\rightarrow\infty$. The state update in DILAND uses these
estimates to compute the local barycentric coordinates, whereas
the update in DLRE uses only the current distance measurements.
 The consistency assumption on the estimates of the inter-sensor distances is quite weak and, as will be shown, is
applicable under practical schemes of estimating inter-sensor
distances through: \begin{inparaenum}[(i)] \item received signal strength~(RSS) and \item
time-of-arrival~(TOA) (see~\cite{patwari_thesis}.)
\end{inparaenum}
We emphasize that DILAND does not require spatial or temporal distributional assumptions on either the communication or the distance measurement noises, except for finiteness of the second order moment, see \textbf{(B.1)}.

Because of {\bf (B.3)}, DLRE \cite{usman_loctsp:08} converges to the exact sensor locations
with a steady state error when the resulting perturbation of the
barycentric coordinates is biased. In contrast, under the new assumption
$\overline{\mbox{\bf (B.3)}}$, DILAND converges to the exact sensor locations
regardless of the bias introduced in the system matrix (at each
iteration) due to noisy distance measurements. This new setup
leads to a behavior and analysis for DILAND that is different from
DLRE's in~\cite{usman_loctsp:08}. This is because using
distance estimates based on the entire past leads to an inherent
strong statistical dependence in the iterative scheme. This dependence makes the
analysis of DILAND different from standard stochastic
approximation arguments, which were used in~\cite{usman_loctsp:08} to prove the convergence properties of DLRE. Furthermore (as we will show in Section~\ref{disc_DILAND}), if we do not have link failures and communication noise, the weight sequence, $\alpha(t)$, in DILAND does not require the square summability condition required by DLRE. Hence, DILAND can be designed to converge faster than DLRE by choosing the DILAND weight sequence to sum faster to infinity than the DLRE weight sequence.

We describe the rest of the paper. Section~\ref{pw} briefly recapitulates our prior work on DILOC and DLRE. We then study the distance estimates and present the main result, DILAND, of this correspondence in Section~\ref{diland_sec}. Section~\ref{sim} presents simulations, and Section~\ref{conc} concludes the paper.

\section{Prior Work}\label{pw}
We briefly recapitulate the setting of the distributed
localization problem, details and discussions are
in~\cite{usman_loctsp:08}.
\subsection{Distributed Localization Algorithm (DILOC)}\label{senm}
Let the row vector~$\mathbf{x}_l(t)\in\mathbb{R}^{1\times m}, l\in\Omega$, be the
$m$-dimensional state that represents the estimated coordinates of sensor~$l$ at time~$t$. Similarly, let the row vector~$\mathbf{u}_k\in\mathbb{R}^{1\times m},~k\in\kappa,$ be the~$m$-dimensional state of the location coordinates of anchor~$k$. DILOC updates at time~$t$ are the following:
\begin{eqnarray}
\mathbf{u}_k(t+1) &=& \mathbf{u}_k(t),\qquad\qquad\qquad\qquad\qquad\qquad~~~ k\in \kappa,\\\label{alg2}
\mathbf{x}_l(t+1) &=& \sum_{n\in{\Omega\cap\Theta_l}}p_{ln} \mathbf{x}_n(t) +
\sum_{k\in{\kappa\cap\Theta_l}}b_{lk} \mathbf{u}_k,\qquad l\in\Omega.
\end{eqnarray}
 The symbol $\Theta_l$ is the triangulation set of sensor $l$, and~$p_{ln}$ and~$b_{lk}$ are the sensor-sensor and the sensor-anchor barycentric
coordinates, \cite{riemann_book}, respectively; these are computed using the inter-node distances in $\mathcal{D}_l$ and the Cayley-Menger determinants~\cite{cayley_men:86} (see~\cite{usman_loctsp:08}.) Let the set of inter-sensor distances required to compute all the barycentric coordinates be ($\ast$~represents the exact value)
\begin{equation}\label{ex_dnn}
\mathbf{d}^\ast =
\{d^\ast_{kn}~|~k,n\in\mathcal{D}_l,l\in\Omega\}.
\end{equation}
The barycentric coordinates lead to the system
matrices~$\mathbf{P}\left(\mathbf{d}^\ast\right) =
\left\{p_{ln}\right\}$ and~$\mathbf{B}\left(\mathbf{d}^\ast\right) =
\left\{b_{lk}\right\}$ of DILOC (we use~$\left(\mathbf{d}^\ast\right)$
to show an implicit dependence on the inter-sensor distances.) We define the matrix $\mathbf{X}(t)$
collecting the row vector states~$\mathbf{x}_l(t)$ for all sensors~$l$
and with column vectors $\mathbf{x}^j(t)$; similarly for~$\mathbf{U}$.
\begin{eqnarray}
\label{eqn:X,xj}
M\times m:\:\mathbf{X}(t) &=& \left[\mathbf{x}_1(t)^T,\ldots,\mathbf{x}_M(t)^T\right]^T
=\left[\mathbf{x}^1(t)\cdots\mathbf{x}^j(t)\cdots\mathbf{x}^m(t)\right]
\\
\label{eqn:U,uj}
(m+1)\times m:\:\mathbf{U} &=& \left[\mathbf{u}_1^T,\ldots,\mathbf{u}_{m+1}^T\right]^T
=\left[\mathbf{u}^1(t)\cdots\mathbf{u}^j(t)\cdots\mathbf{u}^m(t)\right].
\end{eqnarray}
In~(\ref{eqn:X,xj}) and~(\ref{eqn:U,uj}), sub and superindices
indicate row and column vectors of the corresponding matrices, i.e.,
$\mathbf{x}^j(t)$ is the column vector that collects component~$j$
of the row vector state $\mathbf{x}_l(t)$  of all sensors~$l$;
similarly for $\mathbf{u}^j(t)$.

We recall from~\cite{usman_loctsp:08} the following assumptions.

{\bf Structural assumptions}: {\bf (A1)} There are at least~$m+1$ anchors, i.e.,~$|\kappa|=n=m+1$, that do not lie on a hyperplane in~$\mathbb{R}^{m}$; {\bf (A2)} The~$M$ sensors lie inside the convex hull of the anchors, i.e., $\mathcal{C}(\Omega)\subseteq\mathcal{C}(\kappa)$, where $\mathcal{C}(\cdot)$ denotes the convex hull; {\bf (A3)} There exists a triangulation set\footnote{In \cite{usman_loctsp:08}, we give a convex hull inclusion test to identify such triangulation set at each sensor~$l$.}, $\Theta_l\subset\Theta,~\forall~l\in\Omega$, with
$|\Theta_l| = m+1$, such that~$l\in\mathcal{C}(\Theta_l)$;
{\bf (A4)} The sensor~$l$ is assumed to have a communication link,
$n\rightarrow l$, to each~$n\in\Theta_l$ and the inter-node distances
among all the sensors in the set~$\mathcal{D}_l$
are known at sensor~$l$; and
{\bf (A5)} Each anchor,~$k \in \kappa$, has a communication link to at
least one sensor in~$\Omega$.
%
%
%
%
\begin{theorem}[Theorem~1, \cite{usman_loctsp:08}]\label{t_DILOC}
Under \textbf{(A1)-(A5)}, for DILOC~\eqref{alg2}, $\lim_{t\rightarrow\infty}\mathbf{X}(t+1) =
(\mathbf{I-P}(\mathbf{d}^\ast))^{-1}\mathbf{B}(\mathbf{d}^\ast)\mathbf{U}$, i.e., the states converge to the exact sensor locations.
\end{theorem}
\subsection{Distributed Localization in Random Environments (DLRE)}
We start by contrasting assumptions~\textbf{(B.3)} and $\overline{\mbox{\bf (B.3)}}$ (introduced in Section~\ref{intro}).

\textbf{(B.3) Small perturbation of system matrices:} Recall from Section~\ref{intro} that at each sensor~$l$ the distances required to
compute the barycentric coordinates are the inter-node distances
in the set~$\mathcal{D}_l$. In reality, the sensors do not know
the precise distances, $d_{(\cdot)(\cdot)}^\ast$, but estimate the
distances $\widehat{d}_{(\cdot)(\cdot)}(t)$, from RSS or TOA
measurements at time $t$. When we have noisy distance
measurements, we iterate with system
matrices~$\mathbf{P}(\widehat{\mathbf{d}}_t)$ and
$\mathbf{B}(\widehat{\mathbf{d}}_t)$, not with
$\mathbf{P}(\mathbf{d}^\ast)$ and
$\mathbf{B}(\mathbf{d}^\ast)$, where $\widehat{\mathbf{d}}_t $ is the set like $\mathbf{d}^{\ast}$ in~(\ref{ex_dnn}) that collects all required inter-node distance estimates
$\widehat{d}_{kn}(t)$, $k,n\in\mathcal{D}_l,l\in\Omega$.
Assuming the distance measurements are statistically independent
over time, $\mathbf{P}(\widehat{\mathbf{d}}_t)$ and
$\mathbf{B}(\widehat{\mathbf{d}}_t)$ can be written as
\begin{align}\label{rnd3}
\mathbf{P}\left(\widehat{\mathbf{d}}_t\right)&=\mathbf{P}(\mathbf{d}^\ast)
+\mathbf{S}_\mathbf{P}+\widetilde{\mathbf{S}}_\mathbf{P}(t)\triangleq\{\widehat{p}_{ln}(t)\},&\mathbf{B}\left(\widehat{\mathbf{d}}_t\right)=&
\mathbf{B}(\mathbf{d}^\ast)
+\mathbf{S}_\mathbf{B}+\widetilde{\mathbf{S}}_\mathbf{B}(t)\triangleq\{\widehat{b}_{ln}(t)\},&
\end{align}
where~$\mathbf{S}_\mathbf{P}$ and~$\mathbf{S}_\mathbf{B}$ are mean
measurement errors, and
$\{\widetilde{\mathbf{S}}_\mathbf{P}(t)\}_{t\geq0}$ and
$\{\widetilde{\mathbf{S}}_\mathbf{B}(t)\}_{t\geq0}$ are
independent sequence of random matrices with zero-mean and finite
second moments. In particular, even if the distance estimates are unbiased, the
computed~$\mathbf{P}(\widehat{\mathbf{d}}_t)$ and
$\mathbf{B}(\widehat{\mathbf{d}}_t)$ may have non-zero biases,
$\mathbf{S}_\mathbf{P}$ and~$\mathbf{S}_\mathbf{B}$, respectively. The following is the small bias assumption, we made in~\cite{usman_loctsp:08}: $\rho(\mathbf{P}(\mathbf{d}^\ast)+\mathbf{S}_\mathbf{P})<1$. The
DLRE algorithm and its convergence are summarized in the following
theorem. 

\begin{theorem}[Theorem 3, \cite{usman_loctsp:08}]
\label{thm:drle}
Under the noise model \textbf{(B.1)-(B.3)}, the DLRE algorithm given by
\begin{equation}
\label{algass:10}
\mathbf{x}_{l}(t+1)=\left(1-\alpha\left(t\right)\right)\mathbf{x}_{l}(t)+
\alpha(t)\left[\sum_{n\in\Omega\cap\Theta_{l}}
\frac{e_{ln}(t)\widehat{p}_{ln}(t)}{q_{ln}}\left(\mathbf{x}_{n}(t)
+\mathbf{v}_{ln}(t)\right)+\sum_{k\in\kappa\cap\Theta_{l}}
\frac{e_{lk}(t)\widehat{b}_{lk}(t)}{q_{lk}}\left(\mathbf{u}_{k}
+\mathbf{v}_{lk}(t)\right)\right],
\end{equation}
for~$l\in\Omega$,  with $\alpha(t)$ satisfying the persistence condition
$\alpha(t)\geq 0, \sum_t\alpha(t) = \infty, \sum_t\alpha^2(t) < \infty$,
converges almost surely to $\lim_{t\rightarrow\infty}\mathbf{X}(t+1) =
(\mathbf{I-P-S_P})^{-1}(\mathbf{B+S_B})\mathbf{U}(0)$.
\end{theorem}
 The DLRE converges to the exact sensor locations for unbiased random system matrices, i.e., $\mathbf{S_P=S_B=0}$, as established in Theorem~\ref{thm:drle}. As pointed out earlier, even if the distance estimates are unbiased, the system matrices
computed from them may be biased. In such a situation, the DLRE
leads to a nonzero steady state error (bias).

\section{Distributed Sensor Localization with Noisy Distance measurements}\label{diland_sec}
As mentioned, since DLRE at
time~$t$ uses only the current RSS or TOA measurements
to compute distance estimates, the resulting system matrices have
an error bias, i.e., $\mathbf{S_P\neq0}$ and
$\mathbf{S_B\neq0}$. We  use the information from past distance measurements to compute the system matrices
at time~$t$; these become a function of the entire past measurements,
$\{RSS_{s}\}_{s\leq t}$ or $\{TOA_{s}\}_{s\leq t}$. The DILAND algorithm efficiently
utilizes the past information and, as will be shown, leads to a.s.
convergence to the exact sensor locations under practical distance
measurement schemes in sensor networks. To this aim, we review typical models of distance
measurements in wireless settings in Section~\ref{pat_mod} and introduce the DILAND
algorithm in Section~\ref{conv_DILAND}. Section~\ref{disc_DILAND} discusses DILAND.

\subsection{Models for distance measurements}\label{pat_mod}
We explore two standard sensor
networks distance measurements: Received Signal Strength (RSS) and
Time-of-Arrival (TOA). We
borrow experimental and theoretical results from
\cite{patwari_thesis}.

\subsubsection{Received signal strength (RSS)}
In wireless, the signal power decays with a path-loss
exponent, $n_p$, which depends on the environment. If sensor~$a$
sends a packet to sensor~$b$, then $\mbox{RSS}_{ab}$ is the power of the signal
received by sensor~$b$ and the maximum likelihood estimator,
$\widehat{d}_{ab}$, of the distance, $d_{ab}$, between sensors~$a$
and~$b$  is \cite{patwari_thesis}
\begin{equation}\label{rss_est}
\widehat{d}_{ab} = \Delta_0 10^{\frac{\Pi_0 - \mbox{\tiny
RSS}_{ab}}{10n_p}},
\end{equation}
where~$\Pi_0$ is the received power at a short reference distance
$\Delta_0$. For this estimate,
\begin{equation}
\mathbb{E}\left[\widehat{d}_{ab}\right] = Cd_{ab},
\end{equation}
where~$C$ is a multiplicative bias factor. Based on calibration
experiments and a priori knowledge of the environment, we can
obtain precise estimates of~$C$; for typical channels, $C\approx
1.2$ \cite{rappaport_book} and hence scaling~\eqref{rss_est} by~$C$
gives us an unbiased estimate. If the estimate of~$C$ is not
acceptable, we can employ the following scheme.

DILAND is iterative and data is
exchanged at each iteration~$t$. We then have the measurements on~$RSS$
and~$C$ at each iteration. Hence, the distance estimate we employ
is
\begin{equation}\label{sicteen}
\widetilde{d}_{ab}(t) =
\dfrac{\widehat{d}_{ab}(t)}{\widehat{C}(t)},
\end{equation}
where~$\widehat{C}(t)$ is the estimate of~$C$ (for details on this
estimate, see \cite{patwari_thesis} and references therein) at the~$t$-th iteration of
DILAND. Assuming that~$\widehat{d}_{ab}(t)$ and~$\widehat{C}(t)$
are statistically independent (which is reasonable if we use different measurements for these estimates and assume
that the measurement noise is independent over time), we have
\begin{equation}\label{secteen}
\mathbb{E}\left[\widetilde{d}_{ab}(t)\right] = d_{ab}.
\end{equation}

Since at time $t$, we have knowledge of $\{\widetilde{d}_{ab}(s)\}_{0\leq s\leq t}$, we can use the following sequence, $\{\overline{d}_{ab}(t)\}_{t\geq0}$, of distance estimates to compute the barycentric coordinates at time $t$:
\begin{equation}\label{nent}
\overline{d}_{ab}(t) = \dfrac{1}{t}\sum_{s\leq t}\widetilde{d}_{ab}(s)= \dfrac{t-1}{t} \overline{d}_{ab}(t-1) + \dfrac{1}{t}\widetilde{d}_{ab}(t),\qquad \overline{d}_{ab}(0) = \widetilde{d}_{ab}(0).
\end{equation}
Then, from \eqref{sicteen}-\eqref{secteen} and the strong law of large numbers, we have $\overline{d}_{ab}(t)\rightarrow {d}_{ab}$ a.s.~as $t\rightarrow\infty$.

\subsubsection{Time-of-arrival (TOA)}
Time-of-arrival is also used in wireless settings to estimate
distances. TOA is the time for a signal to propagate
from sensor~$a$ to sensor~$b$. To get the distance, TOA is multiplied by
$\nu_p$, the propagation velocity. Over short ranges, the measured time
delay,~$T_{ab}$, can be modeled as a Gaussian
distribution\footnote{Although, as noted before, DILAND does not
require any distributional assumption.} \cite{patwari_thesis} with
mean~$d_{ab}/\nu_p + \mu_T$ and variance~$\sigma_T^2$. The
distance estimate is given by $\widehat{d}_{ab} = (T_{ab} -
\mu_T)\nu_p.$ Based on calibration experiments and a priori
knowledge of the environment, we can obtain precise estimates
of the bias~$\mu_T$; wideband DS-SS measurements \cite{dea:03} have
shown~$\mu_T=10.9$~ns.

Since DILAND is iterative, we can make the required
measurements at each iteration~$t$ and compute an estimate,
$\widehat{\mu}_T(t)$, of the bias, $\mu_T$ (for details on
this computation, see \cite{dea:03}). Then, using the same procedure described for RSS measurements, we can obtain a sequence, $\{\overline{d}_{ab}(t)\}_{t\geq0}$, of distance estimates such that $\overline{d}_{ab}(t)\rightarrow {d}_{ab}$ a.s. as $t\rightarrow\infty$.

In both cases, we note that, if~$\{Z(t)\}_{t\geq 0}$
is a sequence of distance measurements, where $Z=RSS$ or $Z=TOA$,
collected over time, then there exist
estimates~$\overline{\mathbf{d}}_{t}$ with the
property: $\overline{\mathbf{d}}_{t}\rightarrow\mathbf{d}^{\ast}$
a.s. as~$t\rightarrow\infty$. In other words, by a computationally efficient process (e.g., simple averaging) using past distance information, we can
estimate the required inter-sensor distances to arbitrary
precision as~$t\rightarrow\infty$. This leads to the following
natural assumption.

$\overline{\mathbf{(B.3)}}$ \textbf{Noisy distance measurements:}
Let $\{Z(t)\}_{t\geq 0}$ be any sequence of inter-node distance
measurements collected over time. Then, there exists a sequence of
estimates $\{\overline{\mathbf{d}}_{t}\}_{t\geq 0}$ such that,
for all~$t$, $\overline{\mathbf{d}}_{t}$ can be computed
\emph{efficiently} from~$\{X(s)\}_{s\leq t}$ and we have
\begin{equation}
\label{noisy_dist_ass}
\mathbb{P}\left[\lim_{t\rightarrow\infty}\overline{\mathbf{d}}_{t}=
\mathbf{d}^{\ast}\right]=1
\end{equation}

We now present the algorithm DILAND under assumption~$\overline{\mathbf{(B.3)}}$.

\subsection{Algorithm}
\label{conv_DILAND} For clarity of presentation, we analyze
DILAND in the context of noisy distance
measurements only and assume that the inter-sensor communication is perfect
(i.e., no link failures nor communication noise\footnote{In
Subsection~\ref{disc_DILAND}, we discuss the effect of link failures and communication noise.}.)
Let~$\mathbf{P}(\overline{\mathbf{d}}_{t})\triangleq\{\overline{p}_{ln}(t)\}$
and
$\mathbf{B}(\overline{\mathbf{d}}_{t})\triangleq\{\overline{b}_{lk}(t)\}$
be the matrices of barycentric coordinates computed at time~$t$
from the distance estimate~$\overline{\mathbf{d}}_{t}$.
 The DILAND algorithm updates the $j$-th component of the state of all sensors, i.e., updates $\mathbf{x}^j(t)$, the $j$-th component of the location estimates of all sensors, as follows:
 \begin{eqnarray}\label{diland_itm}
\mathbf{x}^j(t+1) &=& (1-\alpha(t))\mathbf{x}^j(t) + \alpha(t)
\left[\mathbf{P}\left(\overline{\mathbf{d}}_t\right)\mathbf{x}^j(t)
+ \mathbf{B}\left(\overline{\mathbf{d}}_t\right)\mathbf{u}^j
\right],\:\:1\leq j\leq m,
\end{eqnarray}
%
%
%
where the weight sequence $\alpha(t)$ satisfies
$\alpha(t)\geq 0$, $\lim_{t\rightarrow\infty}\alpha(t)=0$,
and $\sum_t\alpha(t)=\infty$.
 In particular, here we consider the following choice: for~$a>0$ and~$0<\delta\leq 1$,
\begin{equation}
\label{form_alpha} \alpha(t)=\frac{a}{(t+1)^{\delta}}.
\end{equation}
The update \eqref{diland_itm} is followed by the distance update  \eqref{nent}. The following result gives the convergence
properties of DILAND. The proof is provided in
Appendix~\ref{conv_proofs_DILAND}.
\begin{theorem}
\label{main_DILAND} Assume \textbf{(A.1)-(A.5)} and~$\overline{\mathbf{(B.3)}}$. Let~$\{\mathbf{x}^j(t), 1\leq j\leq m\}_{t\geq 0}$ be the $j$-th coordinate of the state sequence generated by DILAND~(\ref{diland_itm}).
Then $\mathbf{x}^{j}(t)$ converges a.s.~to the
exact $j$-th coordinate location as~$t\rightarrow\infty$, i.e.,
\begin{equation}
\label{main_DILAND100}
\mathbb{P}\left[\lim_{t\rightarrow\infty}\mathbf{x}^j(t)=\left(\mathbf{I-P}
(\mathbf{d}^{\ast})\right)^{-1}\mathbf{B}(\mathbf{d}^{\ast})\mathbf{u}^j,\:\forall
j=1,\ldots,m\right]=1,
\end{equation}
which are the exact sensor locations as established in
Theorem~\ref{t_DILOC}.
\end{theorem}

\subsection{DILAND:Discussions}
\label{disc_DILAND} We discuss the consequences of
Theorem~\ref{main_DILAND}. Unlike DLRE, the proof of
Theorem~\ref{main_DILAND} does not fall under the purview of
standard stochastic approximation. This is because the system
matrices
$\mathbf{P}(\overline{\mathbf{d}}_{t}),\mathbf{B}(\overline{\mathbf{d}}_{t})$ at any
time~$t$ are a function of past distance measurements, making the
sequences of system matrices a strongly dependent sequence. On the
contrary, DLRE assumes that the sequences of system matrices are
independent over time. Thus, DLRE can be analyzed in the framework
of standard stochastic approximation (see, for
example,~\cite{Nevelson}), where it is assumed that the random
perturbations are independent over time (or more generally
martingale difference sequences.) In Appendix~\ref{conv_proofs_DILAND}, we provide a
detailed proof of Theorem~\ref{main_DILAND}, which also provides a
framework for analyzing such stochastic iterative schemes with non
Markovian perturbations.

For clarity of presentation, we ignored the effect of link
failures and additive communication noise in the analysis of
DILAND. In the presence of such effects (i.e.,
assumptions~\textbf{(B.1)-(B.2)}), DILAND can be modified
accordingly and the corresponding state update equation will be
given by: for $1\leq j\leq m$
\begin{equation}
\label{mod_DILAND}
\mathbf{x}^j(t+1) = (1-\alpha(t))\mathbf{x}^j(t) + \alpha(t)
\left[\mathbf{E}\odot\mathbf{P}\left(\overline{\mathbf{d}}_t\right)
\left(\mathbf{x}^j(t)+\mathbf{v}^j(t)\right)
+ \mathbf{E}\odot\mathbf{B}\left(\overline{\mathbf{d}}_t\right)
\left(\mathbf{u}^j+\mathbf{v}^j(t)\right)
\right],\:\:1\leq j\leq m,
\end{equation}
where~$\mathbf{v}^j(t)$ is the measurement noise under~\textbf{(B.1)}, $\odot$ is the Hadamard or pointwise product of two matrices, and
$\mathbf{E}(t)=\left\{\frac{e_{ln}(t)}{q_{ln}}\right\}$ models the link
failures as per~\textbf{(B.2)}, see Section~\ref{intro}.
%
%
%
%
Here we use the notation of Sections~\ref{intro} and~\ref{pw} in the context of
random environments. However, for~\eqref{mod_DILAND}, the weight sequence
$\{\alpha(t)\}$ needs to satisfy the additional square summability
condition, i.e., $\sum_{t\geq 0}\alpha^{2}(t)<\infty$. With this
modification, the results in Theorem~\ref{main_DILAND} will
continue to hold. The proof will incur an additional step, i.e,
after establishing Theorem~\ref{main_DILAND} with no link failures
and noise, a comparison argument will yield the results
(see~\cite{kar-moura-ramanan-IT-2008} for such an argument in the
context of distributed estimation in sensor networks.) In fact, considering a parallel update scheme with state sequence $\{\widetilde{x}(t)\}$ given by
\begin{equation}\label{diland_it100}
\widetilde{\mathbf{x}}^j(t+1) = (1-\alpha(t))\widetilde{\mathbf{x}}^j(t) + \alpha(t)
\left[\mathbf{P}\left(\overline{\mathbf{d}}_t\right)
\widetilde{\mathbf{x}}^j(t)
+ \mathbf{B}\left(\overline{\mathbf{d}}_t\right)
\widetilde{\mathbf{u}}^j\right],\:\:1\leq j\leq m,
\end{equation}
%
%
%
%
we note that the sequence $\{\widetilde{\mathbf{x}}^j(t)\}$ converges a.s. to the exact sensor locations by Theorem~\ref{main_DILAND}. Subtracting this fictitious update scheme from the modified DILAND iterations in eqn.~(\ref{mod_DILAND}), we note that the residual $\widetilde{\mathbf{e}}^j(t)=\mathbf{x}^j(t)-\widetilde{\mathbf{x}}^j(t)$ evolves according to a stable system driven by martingale difference noise (here we note that the effects of link failures and additive channel noise are considered to be temporally independent.) The a.s.~convergence of the sequence $\left\{\widetilde{\mathbf{e}}^j(t)\right\}$ to $\mathbf{0}$ then follows by standard stochastic approximation arguments (see, for example,~\cite{Nevelson}.)

In the absence of link
failures and additive communication noise (possible with a TCP-type protocol and digital communication,) the DILAND needs the
weight sequence~$\alpha(t)$ to satisfy the weaker condition on the weights below \eqref{diland_itm} rather than~$\sum_{t\geq 0}\alpha^{2}(t)<\infty$. On the contrary, the DLRE (even in the absence of link failures
and additive communication noise) requires the additional square
summability condition on the weight sequence, because of the
presence of independent perturbations
$\{\widetilde{\mathbf{S}}_\mathbf{P}(t)\}_{t\geq0}$ and
$\{\widetilde{\mathbf{S}}_\mathbf{B}(t)\}_{t\geq0}$ with non-zero
variance. Thus, DILAND can be designed to converge faster (see also
Section~\ref{simulations} for a numerical example in this regard) than DLRE by choosing a weight sequence that sums to infinity faster than can be achieved in the DLRE. This is because the convergence rate
of such iterative algorithms, essentially, depends on the rate at
which the iteration weights sum to infinity. In particular, the square summability assumption on the weight sequence as required by DLRE (\cite{usman_loctsp:08}) limits the convergence rate to the order of $\frac{1}{\sqrt{t}}$ (\cite{kar-moura-ramanan-IT-2008}), whereas the DILAND does not require square summability of the weight sequence and much better convergence rates can be obtained by making $\{\alpha(t)\}$ sum to $\infty$ faster. However, the requirement that $\alpha(t)\rightarrow 0$ as $t\rightarrow 0$ limits the convergence rate achievable by the DILAND, which is always sub-exponential as exponential convergence rates require the $\alpha(t)$ to be bounded away from zero in the limit.\footnote{The assumption $\alpha(t)\rightarrow 0$ cannot be relaxed in general. This is beacuse, although the distance estimator is assumed to be consistent, at every time $t$, there is an estimation error, and these errors may accumulate in the long run if a constant $\alpha(t)$ was used.}

With respect to triangulation, the variance of the distance measurements may be so large that the convex hull
inclusion test provided in~\cite{usman_loctsp:08} may lead to
incorrect triangulations in the DLRE. However, DILAND refines
the inter-node distance estimates at every time-step, so eventually these estimates get more
accurate and one will end up with the correct triangulation, possibly, with a different set at each iteration. Interested readers can refer to \cite{usman_allerton:08} for a localization algorithm when the triangulation set, $\Theta_l$, is a function of $t$, i.e., $\Theta_l(t)$.

{\bf Remark.} DILAND estimates and updates the inter-node distances at each time step. An alternative is to estimate well the inter-node distances by initially collecting and averaging several measurements and then to run DLRE.  There are clear tradeoffs. The second scheme introduces a set-up phase, akin to a ``training phase,'' implying a delay in getting distance information, it is  sensitive to variations in sensor locations, and it will show a remnant bias because it stops the estimation of the inter-node distances after a finite number of iterations, which means there will be a residual error in the inter-sensor distances. This residual error induces a bias in the distance estimates and so in the barycentric coordinates, which means a residual error in the sensor coordinates, no matter how long DLRE is run. DILAND does not suffer from these limitations, it has no initial delay, it is robust to variations in the sensors geometries since it can sustain these to a certain degree given its adaptive nature, and it continues reducing the asymptotic error since the inter-sensor distances are continuously updated. On the other hand, DILAND's iterations are more onerous since it requires continuous updating of the inter-node distances and recomputing the barycentric coordinates. Very similar tradeoffs arise an are studied in the context of standard consensus in~\cite{karmoura-tspjan09}
\section{Simulations}\label{sim}
\label{simulations} In this section, we compare DLRE
\cite{usman_loctsp:08} with DILAND. In the first simulation, we do
not include link failures and communication noise and focus only on
distance noise. This study presents the advantage of choosing the relaxed
weight sequence, $\alpha_{\mbox{\tiny DILAND}}(t)$ below \eqref{diland_itm},
 over $\alpha_{\mbox{\tiny DLRE}}$ in Theorem~\ref{thm:drle}.
 We assume that the
distance, $d_{lj}$, between any two sensors~$l$ and~$j$ is
corrupted by additive Gaussian noise, i.e.,
\begin{equation}
\widetilde{d}_{lj}(t) = d_{lj}^\ast + w_{lj}(t),
\end{equation}
where~$w_{lj}\sim\mathcal{N}\left(0,\frac{d_{lj}}{10}\right)$.
We choose the variance of the noise in distance measurement as
$10\%$ of the actual distance. We adopt the estimator~\eqref{nent}, which is an estimator at time~$t$ given by the average of the past observations.
 Hence, \eqref{nent} becomes our distance update. It is
straightforward to show that~$\overline{d}_{lj}(t)\rightarrow
d_{lj}^\ast$. We use the following weight sequences for DLRE and DILAND compared in Fig.~1 (center):
\begin{align}
\alpha_{\mbox{\tiny DLRE}}(t)
&=\dfrac{1}{t^{0.55}},&\alpha_{\mbox{\tiny DILAND}}(t)
&=\dfrac{1}{t^{0.25}}.&
\end{align}
These satisfy the persistence condition in Theorem~\ref{thm:drle}
 and below \eqref{diland_itm}, 
  respectively. In particular, as noted before, $\alpha_{\mbox{\tiny DILAND}}$ does not have to satisfy~$\sum_{t\geq 0}\alpha_{\mbox{\tiny DILAND}}^{2}(t)<\infty$, which leads to faster convergence (see Fig.~1 (right.))

We simulate an~$N=50$ node network in~$m=2$-D Euclidean space. We have~$m+1=3$ anchors (with known locations) and~$M=47$
sensors (with unknown locations). This network (and
appropriate triangulations) is shown in Fig.~1 (left), where
$\nabla$s represent the anchors and~$\circ$s represent the
sensors. Fig.~1 (center) shows the weight sequences chosen for DLRE and DILAND, whereas Fig.~1 (right) shows the normalized mean
squared error
\begin{equation}
\mbox{MSE}_t = \dfrac{\dfrac{1}{M}\sum_{l=1}^M\sum_{j=1}^m
\left(x_{l}^{j}(t)-x_{l}^{{j}\ast}\right)^2}{\dfrac{1}{M}\sum_{l=1}^M\sum_{j=1}^m
\left(x_{l}^{j}(0)-x_{l}^{{j}\ast}\right)^2}.
\end{equation}
\begin{figure}
\centering
    \includegraphics[width=1.9in]{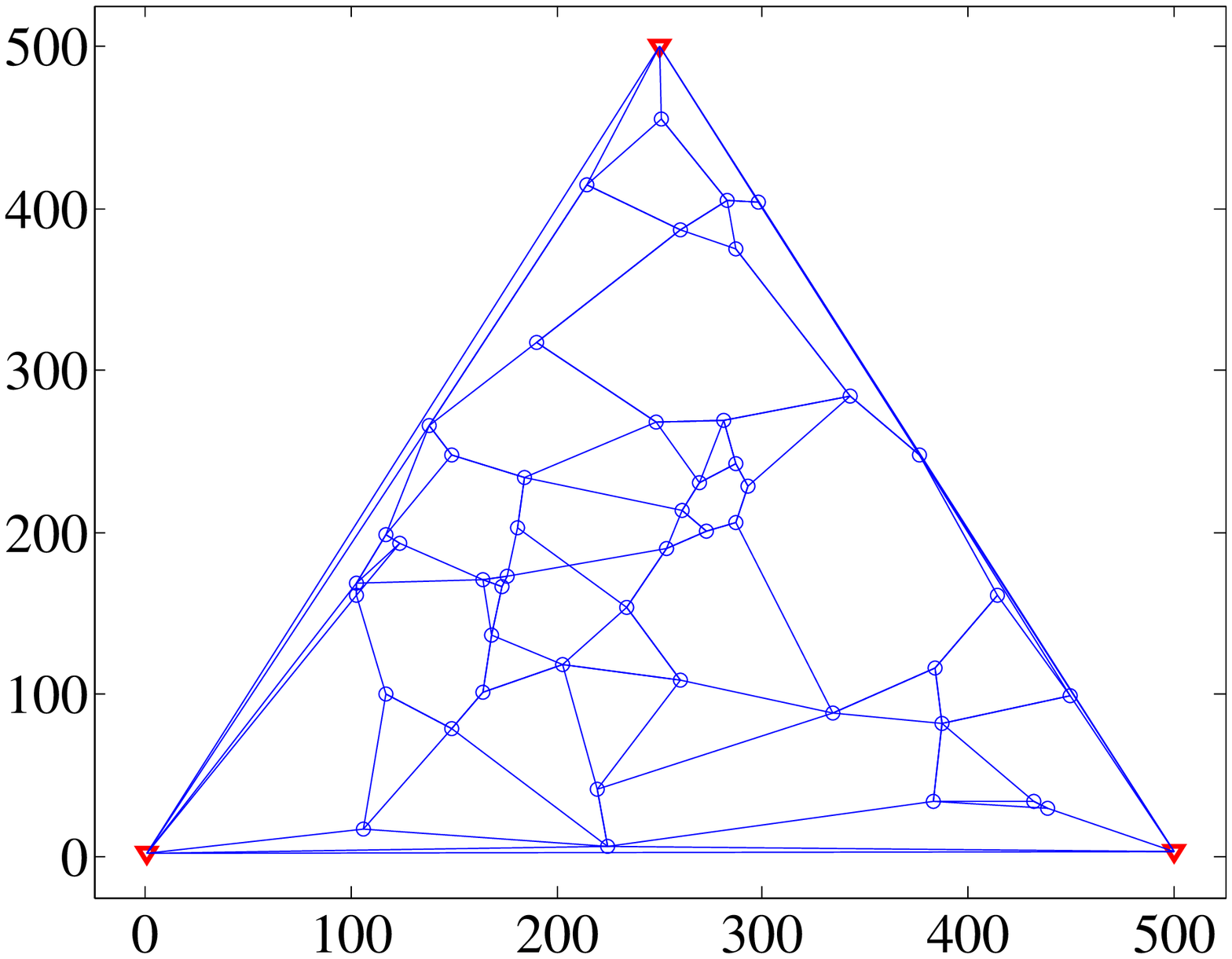}
 \hspace{.1cm}
   \includegraphics[width=1.9in]{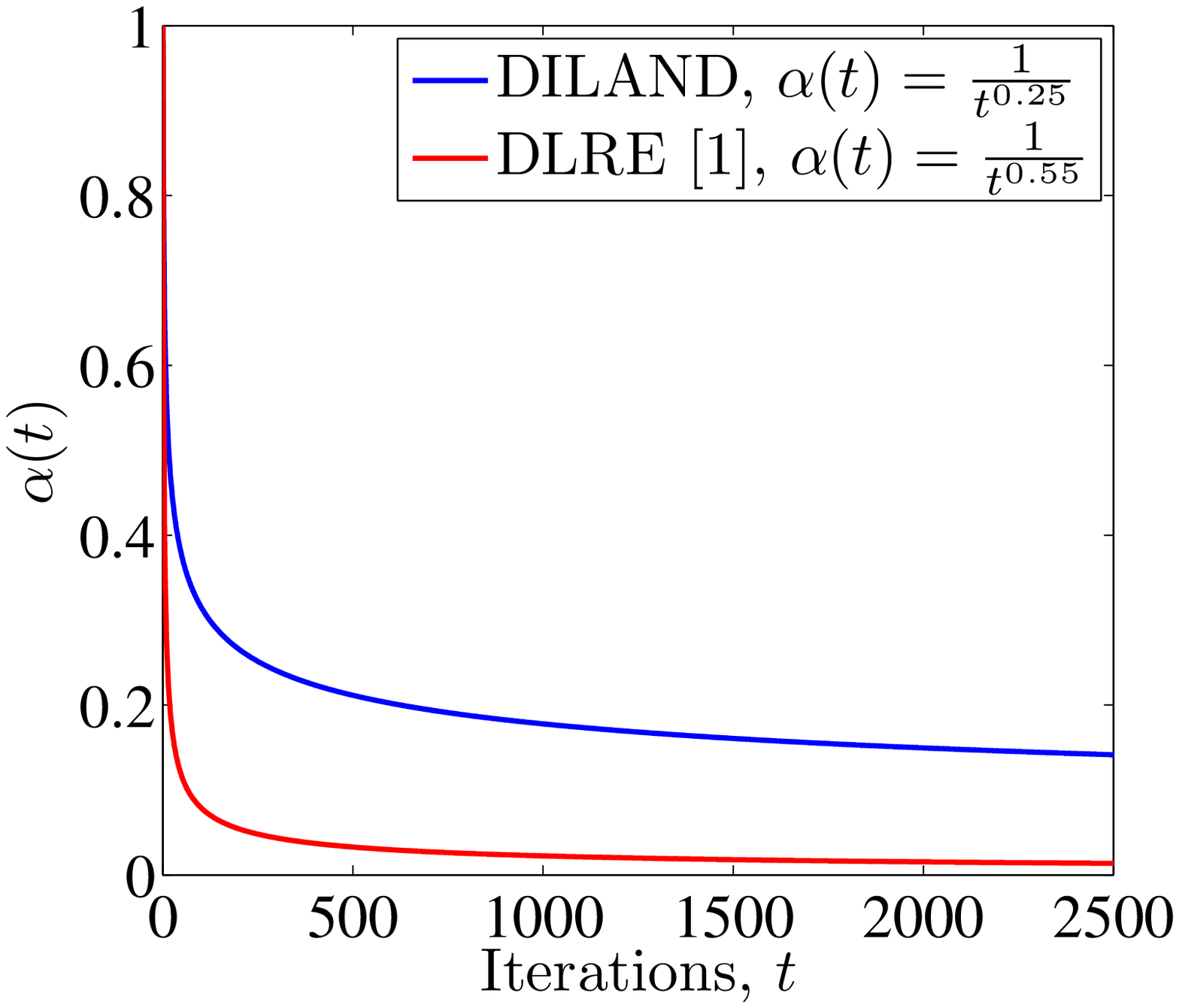}
\hspace{.1cm}
    \includegraphics[width=1.9in]{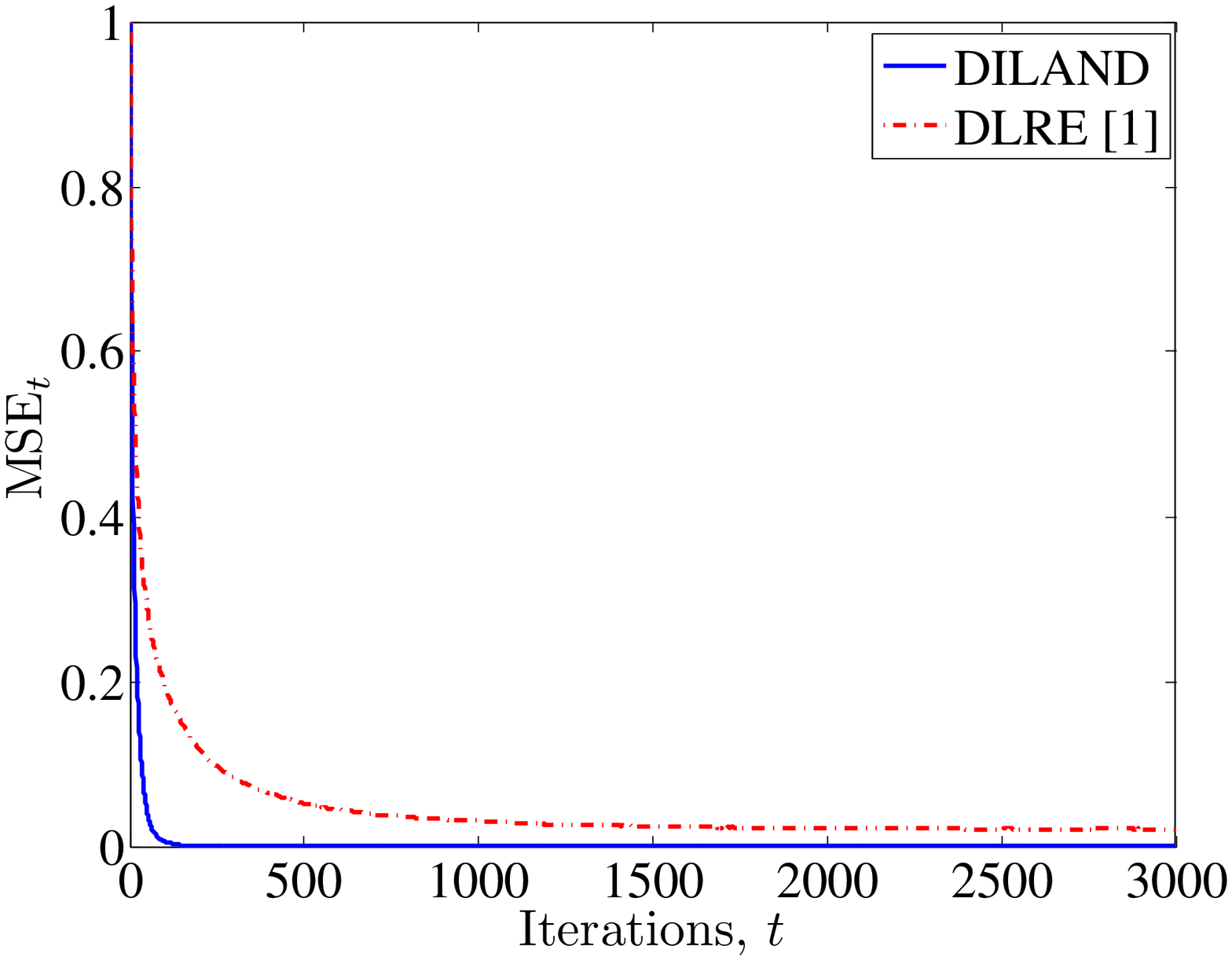}
 \caption{Distance noise only: Left: $N=50$ node network, $m=2$-D space, $m+1=3$, anchors and $M=47$ sensors. Center: Decreasing weight sequences, $\alpha(t)$, for DLRE and DILAND. Right: Normalized mean squared error for DLRE and DILAND.} \label{saa}
\end{figure}

In our second simulation, we keep the experimental setup, but include a zero-mean Gaussian random variable with unit variance as communication noise. We further assume that the communication links are active $90\%$ of the time. With this additional randomness in the localization algorithm, we cannot use the relaxed DILAND weights, but use the same weights given in Theorem~\ref{thm:drle}
 for both DLRE and DILAND, see Section~\ref{disc_DILAND} on a discussion on this choice. The network (along with
appropriate triangulations) is shown in Fig.~2 (left.) Fig.~2 (center) shows the weights sequences chosen for DLRE and DILAND, whereas Fig.~2 (right) shows the normalized mean
squared error.
\begin{figure}
\centering
    \includegraphics[width=1.9in]{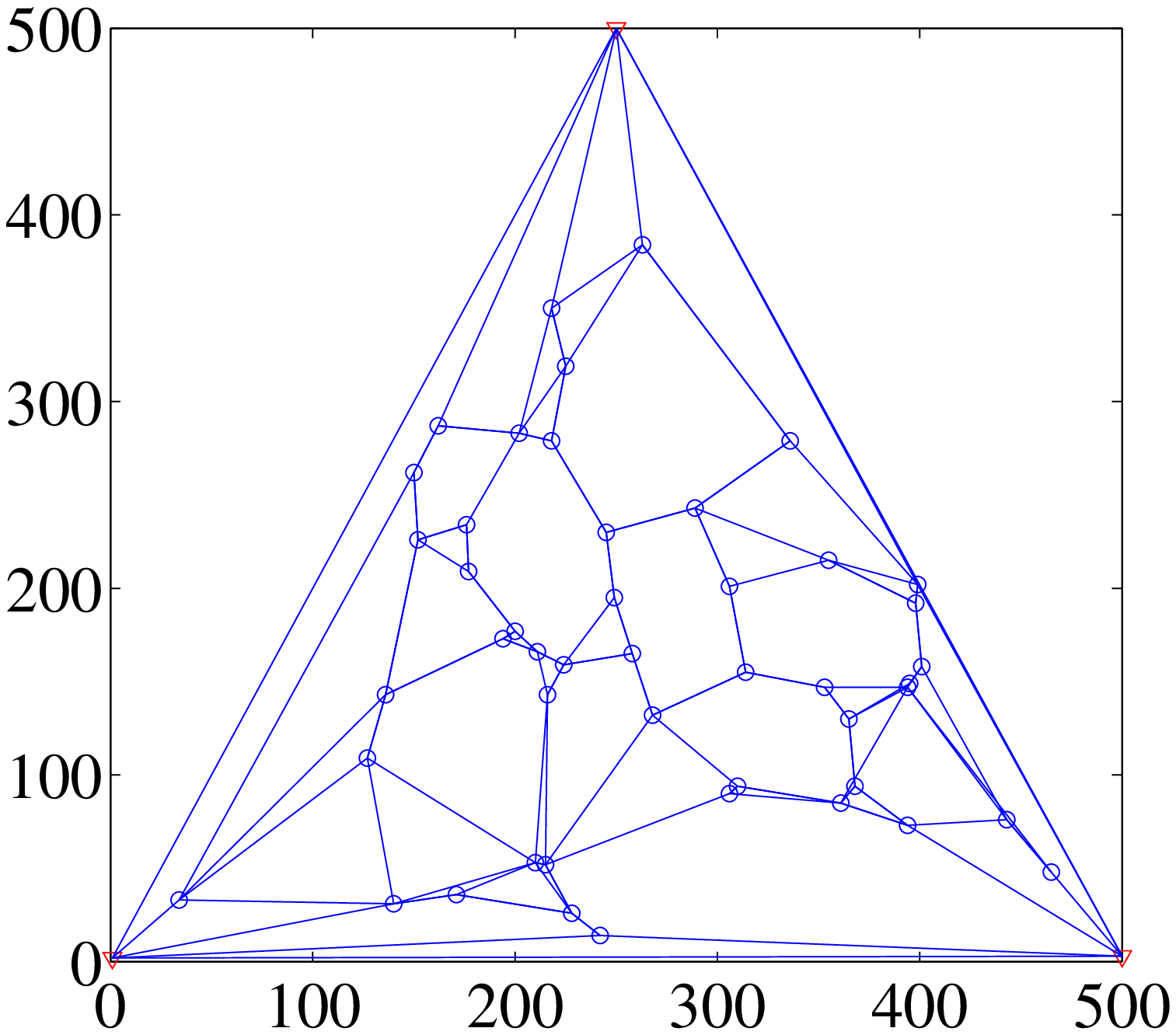}
\hspace{.1cm}
    \includegraphics[width=1.9in]{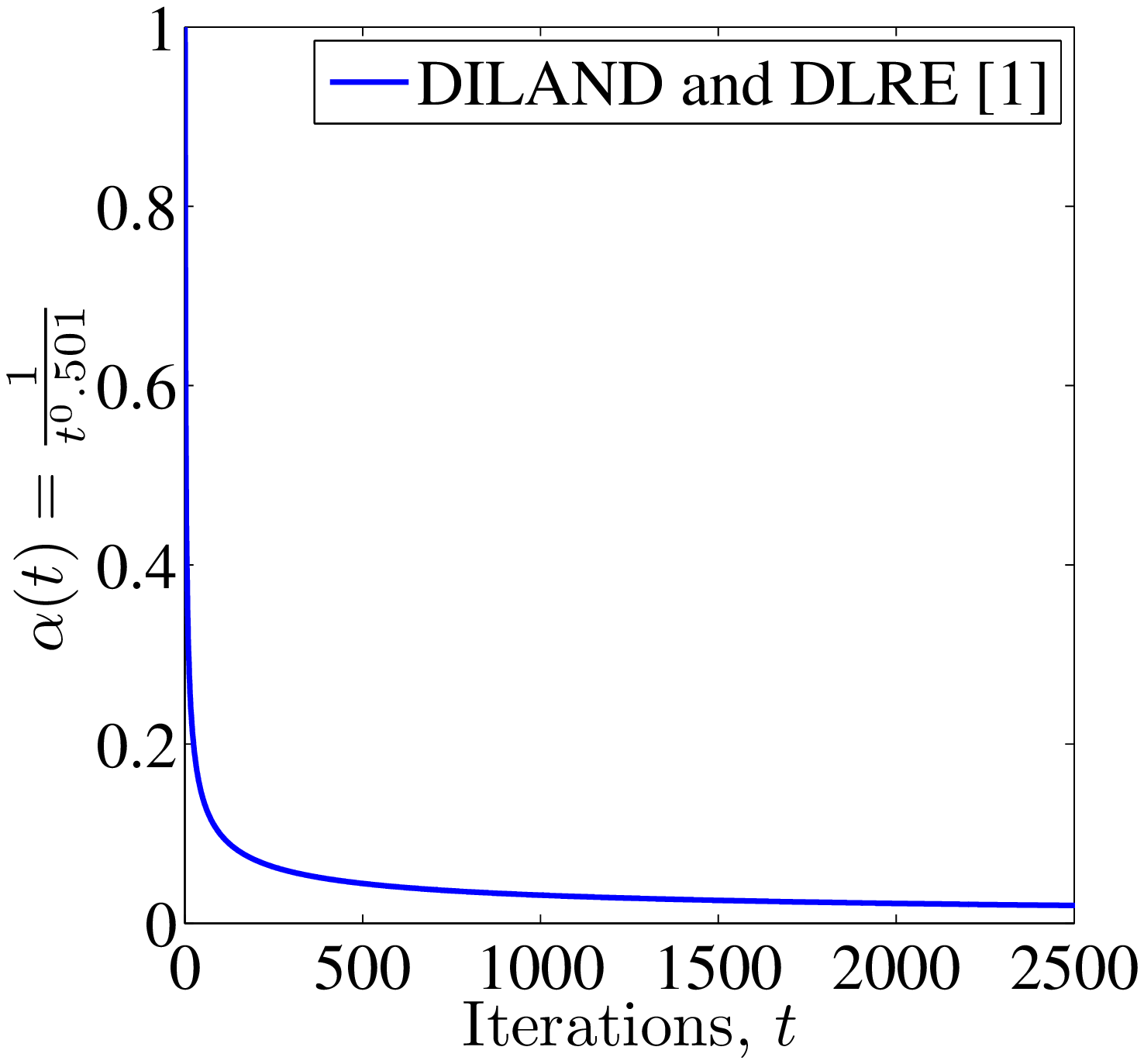}
     \hspace{.1cm}
    \includegraphics[width=1.9in]{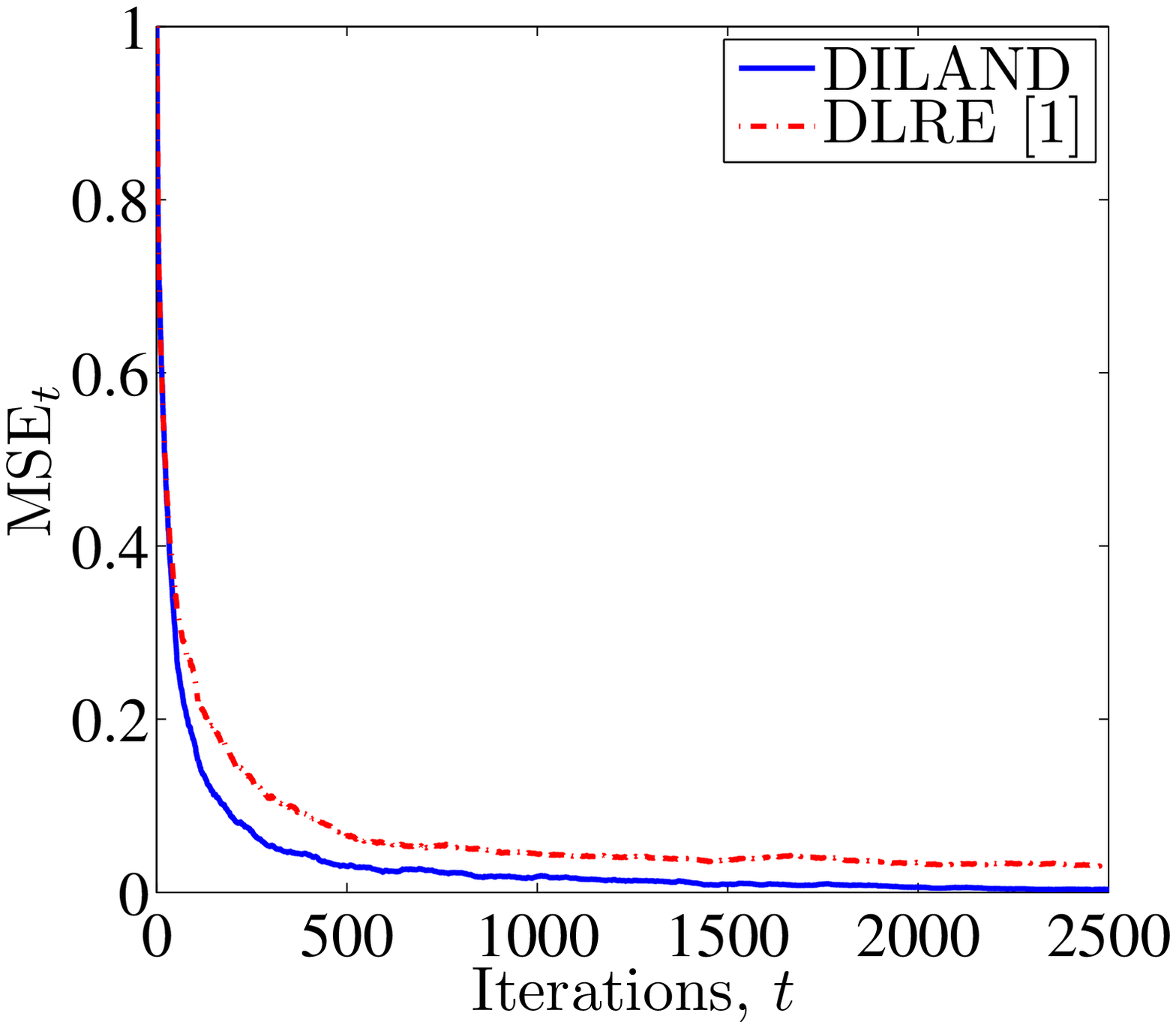}
\caption{Distance noise, communication noise, and link failure: Left: $N=50$ node network, $m=2$-D space, $m+1=3$ anchors, and $M=47$ sensors. Center: Same weight sequence, $\alpha(t)$, for DLRE and DILAND. Right: Normalized mean squared error for DLRE and DILAND.} \label{saa-2}
\end{figure}
In both studies, we note that DLRE converges with a steady state error as is
emphasized earlier, whereas, DILAND converges
almost surely to the exact sensor locations. Furthermore, the
simulations confirm that DILAND converges faster (faster convergence in the presence of all random phenomena is a consequence of refined distance measurements) because the
square summability condition on~$\alpha_{\mbox{\tiny DLRE}}(t)$ as required by DLRE is not
required by~$\alpha_{\mbox{\tiny DILAND}}(t)$.

\section{Conclusions}\label{conc}
\label{conclusions} In this correspondence, we present a distributed
localization algorithm, DILAND, that converges a.s. to the exact
sensor locations when we have (i) communication noise, (ii) random
link failures, and (iii) noisy distance measurements. We build on
our earlier work, DLRE \cite{usman_loctsp:08}, that assumes that the
noise on distance measurements results into a small signal
perturbation of the system matrices; this perturbation is biased,
in general. Due to this bias, there is a non-zero steady error in
the location estimates of DLRE.
With the additional assumption that the distance estimates at each
iteration of DILAND converge a.s. to the exact distances, we show that
the steady state error in the location estimates of DILAND is
zero. We further show that if there is no communication noise and
link failures, possible in some communication environments, the
convergence rate of DILAND is faster than the convergence rate of
DLRE. We provide simulations to assert the analytical results.

\vspace*{-.5cm}
\appendices
%
%
%
%
%
%
%

\section{Convergence of DILAND}
\label{imp_res}
{\small We first summarize a relevant result that is
needed for proving Theorem~\ref{main_DILAND}. The proof of this Lemma is in~\cite{kar-moura-ramanan-IT-2008}.
\begin{lem}[Lemma 18, \cite{kar-moura-ramanan-IT-2008}]
\label{prop} Let the sequences~$\{r_{1}(t)\}_{t\geq 0}$ and
$\{r_{2}(t)\}_{t\geq 0}$ be given by
\begin{equation}
\label{prop1}
r_{1}(t)=\frac{a_{1}}{(t+1)^{\delta_{1}}},~~~~~r_{2}(t)=\frac{a_{2}}{(t+1)^{\delta_{2}}}
\end{equation}
where~$a_{1},a_{2},\delta_{2}\geq 0$ and~$0\leq\delta_{1}\leq 1$.
Then, if~$\delta_{1}=\delta_{2}$, there exists~$K>0$, such that,
for non-negative integers,~$s<t$,
\begin{equation}
\label{prop2}
0\leq\sum_{k=s}^{t-1}\left[\prod_{l=k+1}^{t-1}(1-r_{1}(l))\right]r_{2}(k)\leq
K
\end{equation}
Moreover, the constant~$K$ can be chosen independently of~$s,t$.
 Also, if~$\delta_{1}<\delta_{2}$, then, for arbitrary fixed~$s$,
\begin{equation}
\label{prop3}
\lim_{t\rightarrow\infty}\sum_{k=s}^{t-1}\left[\prod_{l=k+1}^{t-1}(1-r_{1}(l))\right]r_{2}(k)=0
\end{equation}
\end{lem}
%
\label{conv_proofs_DILAND} {\small

We now prove the convergence of DILAND given by~\eqref{diland_itm}. Unless
otherwise noted, the norm~$\left\|\cdot\right\|$ refers to the
standard Euclidean 2-norm. 
 The following lemma shows that the DILAND
iterations are bounded for all~$t$.
\begin{lem}
\label{bound} Consider the sequence of iterations in
\eqref{diland_itm}. We have
\begin{equation}
\label{bound2} \mathbb{P}\left[\sup_{t\geq
0}\left\|\mathbf{x}^{j}(t)\right\|<\infty,~~1\leq j\leq m\right]=1.
\end{equation}
In other words, the sequence~$\{\mathbf{x}^{j}(t)\}_{t\geq 0}$
remains bounded a.s. for all $j$.
\end{lem}
\begin{proof}
We rewrite~(\ref{diland_itm}) as
\begin{eqnarray}
\mathbf{x}^j(t+1)&=& (1-\alpha(t))\mathbf{x}^j(t) + \alpha(t) \left[
\left(\mathbf{P}\left(\overline{\mathbf{d}}_t\right)+
\mathbf{P}\left(\mathbf{d}^\ast\right)
-\mathbf{P}\left(\mathbf{d}^\ast\right)\right) \mathbf{x}^j(t) +
\mathbf{B}\left(\overline{\mathbf{d}}_t\right)\mathbf{u}^j \right],\\
\label{bound5}
&=&\left[(1-\alpha(t)I)+\alpha(t)\mathbf{P}(\mathbf{d}^{\ast})\right]\mathbf{x}^{j}(t)+\alpha(t)
\left[\mathbf{P}
(\overline{\mathbf{d}}_{t})-\mathbf{P}(\mathbf{d}^{\ast})\right]
\mathbf{x}^{j}(t)+\alpha(t)\mathbf{B}(\overline{\mathbf{d}}_{t})\mathbf{u}
\end{eqnarray}
Since~$\rho\left(\mathbf{P}(\mathbf{d}^{\ast})\right)<1$ (recall
Theorem~\ref{t_DILOC}), it follows from the properties of linear
operators on Banach spaces that there exists a norm
$\left\|\cdot\right\|_{P}$ such that the corresponding induced
norm of the linear operator~$\mathbf{P}(\mathbf{d}^{\ast})$, satisfies
\begin{equation}
\label{bound6}
\left\|\mathbf{P}(\mathbf{d}^{\ast})\right\|_{P}=\max_{\mathbf{y}\neq\mathbf{0}}
\frac{\left\|\mathbf{P}(\mathbf{d}^{\ast})\mathbf{y}\right\|_{P}}{\left\|\mathbf{y}\right\|_{P}}<1
\end{equation}
Moreover, such a norm can be chosen to be equivalent to the
Euclidean norm~$\left\|\cdot\right\|$ (see, for
example,~\cite{Rodrigues}), i.e., there exist constants
$c_{1},c_{2}>0$, such that,
$c_{1}\left\|\cdot\right\|\leq\left\|\cdot\right\|_{P}\leq
c_{2}\left\|\cdot\right\|$. In other words,~$\left\|\cdot\right\|$
and~$\left\|\cdot\right\|_{P}$ generate the same topology on the
Euclidean space.
Since~$\alpha(t)\rightarrow 0$, we can choose~$t_{0}$ sufficiently
large, such that,~$\alpha(t)<1,~~\forall t\geq t_{0}$.
From the above construction, we have for, $t\geq t_{0}$,
\begin{equation}
\label{bound9}
\left\|(1-\alpha(t)\mathbf{I})+\alpha(t)\mathbf{P}(\mathbf{d}^{\ast})\right\|_{P}
\leq
1-\alpha(t)+\alpha(t)\left\|\mathbf{P}(\mathbf{d}^{\ast})\right\|_{P}
= 1-\lambda^{\ast}\alpha(t)
\end{equation}
where~$0<\lambda^{\ast}\leq 1$.
From~$\overline{\mathbf{(B.3)}}$, we recall
$\overline{\mathbf{d}}_{t}\rightarrow\mathbf{d}^{\ast}$ a.s. and
from~$\mathbf{P}(\cdot), \mathbf{B}(\cdot)$ being
continuous functions of~$\mathbf{d}_t$, we have
\begin{equation}
\label{bound12} \mathbf{P}(\overline{\mathbf{d}}_{t})\rightarrow
\mathbf{P}(\mathbf{d}^{\ast}),~~~\mbox{a.s.}~~~~~~~~~~\mathbf{B}(\overline{\mathbf{d}}_{t})\rightarrow
\mathbf{B}(\mathbf{d}^{\ast}),~~~\mbox{a.s.}
\end{equation}
Now fix a sample path~$\omega$. There exists~$t_{1}(\omega)$
sufficiently large, such that, for~$0<\varepsilon<\lambda^{\ast}$,
we have, if~$t\geq t_{1}(\omega)$
\begin{equation}
\label{bound14}
\left\|\mathbf{P}(\overline{\mathbf{d}}_{t})-\mathbf{P}(\mathbf{d}^{\ast})\right\|_{P}\leq
\varepsilon
\end{equation}
Also, the a.s. convergence of
$\mathbf{B}(\overline{\mathbf{d}}_{t})$, implies there exists
$\lambda_{1}(\omega)$ such that
$\left\|B(\overline{\mathbf{d}}_{t})\right\|_{P}\leq
\lambda_{1}(\omega)$. Then, for~$t\geq \max(t_{0},t_{1}(\omega))$
\begin{eqnarray}
\label{bound15} \left\|\mathbf{x}^j(t+1)\right\|_{P} & \leq &
\left(1-\lambda^{\ast}\right)\left\|\mathbf{x}^{j}(t)\right\|_{P}+\varepsilon\alpha(t)\left\|\mathbf{x}^{j}(t)\right\|_{P}+\alpha(t)\lambda_{1}(\omega)\left\|\mathbf{u}^{j}\right\|_{P}\nonumber
\\ & = &
\left(1-(\lambda^{\ast}-\varepsilon)\alpha(t)\right)\left\|\mathbf{x}^{j}(t)\right\|_{P}+\alpha(t)\lambda_{1}(\omega)\left\|\mathbf{u}^{j}\right\|_{P}
\end{eqnarray}
Let~$t_{2}(\omega)=\max(t_{0},t_{1}(\omega))$,
$a_{1}=\lambda^{\ast}-\varepsilon$, and
$a_{2}(\omega)=\lambda_{1}(\omega)\left\|\mathbf{u}^{j}\right\|_{P}$.
Continuing the above recursion, we have for
$t\geq t_2(\omega)$ 
\begin{eqnarray}
\label{bound22} \left\|\mathbf{x}^{j}(t)\right\|_{P} & \leq &
\left(\prod_{k=t_{2}(\omega)}^{t-1}(1-a_{1}\alpha(t))\right)\left\|\mathbf{x}(t_{2}(\omega))\right\|_{P}+\sum_{k=t_{2}(\omega)}^{t-1}\left[\left(\prod_{l=k+1}^{t-1}(1-a_{1}\alpha(l))\right)a_{2}\alpha(k)\right]\nonumber
\\ & \leq &
\left\|\mathbf{x}(t_{2}(\omega))\right\|_{P}+\sum_{k=t_{2}(\omega}^{t-1}\left[\left(\prod_{l=k+1}^{t-1}(1-a_{1}\alpha(l))\right)a_{2}(\omega)\alpha(k)\right]
\end{eqnarray}
The second term falls under the purview of Lemma~\ref{prop} in
Appendix~\ref{imp_res} with~$\delta_{1}=\delta_{2}=\delta$, and we
have
\begin{equation}\label{bound113}
\sup_{t\geq 0}\left\|\mathbf{x}^{j}(t)\right\|_{P}\leq K(\omega)
\end{equation}
for some~$K(\omega)>0$. Since the above analysis holds a.s., we have
$\mathbb{P}\left[\sup_{t\geq
0}\left\|\mathbf{x}^{j}(t)\right\|_{P}<\infty\right]=1$.
The lemma then follows from the fact that the norms
$\left\|\cdot\right\|$ and~$\left\|\cdot\right\|_{P}$ are
equivalent. In particular, boundedness in~$\left\|\cdot\right\|$
is equivalent to boundedness in~$\left\|\cdot\right\|_{P}$.
\end{proof}

We now present the proof of Theorem~\ref{main_DILAND}.

\begin{proof}[Proof of Theorem~\ref{main_DILAND}]
We use a comparison argument. To this end, consider the idealized
update
\begin{equation}
\label{main3}
\widetilde{\mathbf{x}}^j(t+1)=\left(1-\alpha(t)\right)\widetilde{\mathbf{x}}^j(t)+\alpha(t)
\left[\mathbf{P}(\mathbf{d}^{\ast})\widetilde{\mathbf{x}}(t)+
\mathbf{B}(\mathbf{d}^{\ast})\mathbf{u}^j\right]
\end{equation}
It follows from Theorem~\ref{t_DILOC} that, for~$j=1,\ldots,m$,
\begin{equation}
\label{main4}
\lim_{t\rightarrow\infty}\widetilde{\mathbf{x}}^j(t)=\left(\mathbf{I-P}(\mathbf{d}^{\ast})\right)
^{-1}\mathbf{B}(\mathbf{d}^{\ast})\mathbf{u}^j.
\end{equation}
Define the sequence~$\{\mathbf{e}^j(t)=\mathbf{x}^j(t)-
\widetilde{\mathbf{x}}^j(t)\}_{t\geq 0}$.
 We then have
\begin{equation}
\label{main6}
\mathbf{e}^j(t+1)=\left(1-\alpha(t)\right)\mathbf{e}^j(t)+\alpha(t)\mathbf{P}
(\mathbf{d}^{\ast})\mathbf{e}^j(t)+\alpha(t)\left(\mathbf{P}(\overline{\mathbf{d}}_{t})-
\mathbf{P}(\mathbf{d}^{\ast})\right)\mathbf{x}^j(t)+\alpha(t)\left(\mathbf{B}(\overline{\mathbf{d}}_{t})-
\mathbf{P}(\mathbf{d}^{\ast})\right)\mathbf{u}^j.
\end{equation}
Recall the norm~$\left\|\cdot\right\|_{P}$ introduced in the proof
of Lemma~\ref{bound}. From~(\ref{bound9}), we have for~$t\geq
t_{0}$,
\begin{equation}
\label{main7}
\left\|(1-\alpha(t)\mathbf{I})+\alpha(t)\mathbf{P}(\mathbf{d}^{\ast})\right\|_{P}
\leq 1-\lambda^{\ast}\alpha(t)
\end{equation}
where~$0<\lambda^{\ast}\leq 1$. Now, fix a sample path ($\omega$). Consider~$\tau>0$. By
\eqref{bound12}, there exists sufficiently
large~$t_{3}(\omega,\tau)$ such that 
\begin{equation}
\label{main8}
\forall t\geq t_{3}(\omega,\tau):\:\:\left\|\mathbf{B}(\overline{\mathbf{d}}_{t})-
\mathbf{B}(\mathbf{d}^{\ast})\right\|_{P}\leq\tau,~~~~~~~\left\|
\mathbf{P}(\overline{\mathbf{d}}_{t})-\mathbf{P}(\mathbf{d}^{\ast})
\right\|_{P}\leq\tau
\end{equation}
Define~$t_{4}(\omega,\tau)=\max_{t_{0},t_{3}(\omega,\tau)}$.
We then have for~$t\geq t_{4}(\omega,\tau)$
\begin{eqnarray}
\label{main10} \left\|\mathbf{e}^j(t+1)\right\|_{P} & \leq &
\left\|(1-\alpha(t)I)+\alpha(t)\mathbf{P}(\mathbf{d}^{\ast})\right\|_{P}\left\|\mathbf{e}^j(t)\right\|_{P}\\
&+&\alpha(t)\left\|\mathbf{P}(\overline{\mathbf{d}}_{t})-\mathbf{P}
(\mathbf{d}^{\ast})\right\|_{P}\left\|\mathbf{x}^j(t)
\right\|_{P}+\alpha(t)\left\|\mathbf{B}(\overline{\mathbf{d}}_{t})-\mathbf{B}(\mathbf{d}^{\ast})\right\|_{P}\left\|
\mathbf{u}^j\right\|_{P} \nonumber \\ & \leq &
\left(1-\lambda^{\ast}\alpha(t)\right)\left\|\mathbf{e}^j(t)\right\|_{P}+\tau\alpha(t)
K(\omega)+\tau\left\|\mathbf{u}^j\right\|_{P}\nonumber \\ & \leq &
\left(1-\lambda^{\ast}\alpha(t)\right)\left\|\mathbf{e}^j(t)\right\|_{P}+\tau\alpha(t)\left(K(\omega)
+\left\|\mathbf{u}^j\right\|_{P}\right)
\end{eqnarray}
where~$K(\omega)$ is defined in~(\ref{bound113}). Continuing the
above recursion, we have for~$t>t_{4}(\omega,\tau)$
\begin{equation}
\label{main11} \left\|\mathbf{e}^j(t)\right\|_{P}\leq
\left(\prod_{k=t_{4}(\omega,\tau)}^{t-1}(1-\lambda^{\ast}\alpha(t))\right)\left\|\mathbf{e}^j(t_{4}(\omega,\tau))\right\|_{P}
+\tau\sum_{k=t_{4}(\omega,\tau)}^{t-1}\left[\left(\prod_{l=k+1}^{t-1}(1-\lambda^{\ast}\alpha(l))\right)\left(K(\omega)+\left\|\mathbf{u}^j\right\|_{P}\right)\alpha(k)\right]
\end{equation}
Using the inequality~$1-\lambda^{\ast}\alpha(t)\leq
e^{-\lambda^{\ast}\alpha(t)}$ for sufficiently small
$\lambda^{\ast}\alpha(t)$, we have
\begin{equation}
\label{main13}
\lim_{t\rightarrow\infty}\left(\prod_{k=t_{4}(\omega,\tau)}^{t-1}(1-\lambda^{\ast}\alpha(t))\right)\left\|\mathbf{e}^j(t_{4}(\omega,\tau))\right\|_{P}
 \leq
\lim_{t\rightarrow\infty}e^{\left(-\lambda^{\ast}\sum_{k=t_{4}(\omega,\tau)}^{t-1}\alpha(t)\right)}\left\|\mathbf{e}^j(t_{4}(\omega,\tau))\right\|_{P}
 =  0
\end{equation}
The last step follows from the fact that
$\sum_{k=t_{4}(\omega,\tau)}^{\infty}\alpha(t)=\infty$.
From Lemma~\ref{prop}, we have
\begin{equation}
\label{main15}
\sum_{k=t_{4}(\omega,\tau)}^{t-1}\left[\left(\prod_{l=k+1}^{t-1}(1-\lambda^{\ast}\alpha(l))\right)\left(K(\omega)+\left\|\mathbf{u}^j\right\|_{P}\right)\alpha(k)\right]\leq
c_{4}(\omega)
\end{equation}
Note, in particular,~$c_{4}(\omega)$ is independent of~$\tau$.
 We then have from~(\ref{main11})
\begin{eqnarray}
\label{main16}
\limsup_{t\rightarrow\infty}\left\|\mathbf{e}^j(t)\right\|_{P}&\leq&\nonumber
\limsup_{t\rightarrow\infty}\left(\prod_{k=t_{4}(\omega,\tau)}^{t-1}(1-\lambda^{\ast}\alpha(t))\right)\left\|\mathbf{e}^j(t_{4}(\omega,\tau))\right\|_{P}\\
&+&\tau\limsup_{t\rightarrow\infty}\sum_{k=t_{4}(\omega,\tau)}^{t-1}\left[\left(\prod_{l=k+1}^{t-1}
(1-\lambda^{\ast}\alpha(l))\right)\left(K(\omega)+\left\|\mathbf{u}^j\right\|_{P}\right)\alpha(k)\right]
\nonumber
 \leq  \tau c_{4}(\omega).
\end{eqnarray}
Since~(\ref{main16}) holds for arbitrary~$\tau>0$, we have
$
\lim_{t\rightarrow\infty}\left\|\mathbf{e}^j(t)\right\|_{P}=0.
$
The theorem then follows from the
fact that this convergence to zero 
 holds for~$\omega$ a.s and from the equivalence of the norms
$\left\|\cdot\right\|$ and~$\left\|\cdot\right\|_{P}$.
\end{proof}

}

\vspace*{-.25cm}

\bibliographystyle{IEEEbib}
\bibliography{IEEEabrv,globalRef}

\end{document}